\documentclass{llncs}
\usepackage{llncsdoc}
\usepackage{graphicx}
\usepackage{epsf}
\usepackage{epsfig}

\usepackage{amsmath}
\usepackage{amssymb}
\usepackage{cite}
\usepackage{color}
\usepackage{makeidx}
\usepackage{hyperref}
\usepackage{url}

\DeclareMathOperator*{\argmin}{arg\,min}
\newcommand{\bs}{\backslash}

\newcommand{\ignore}[1]{}

\newcommand{\citep}[1]{\cite{#1}}
\newcommand{\OPT}{{\rm OPT}}

\newcommand{\OPTHS}{{\rm HS}}

\newcommand{\DEN}{{\rm DEN}}

\newcommand{\ETD}{{\rm ETD}}
\newcommand{\SETD}{{\rm SETD}}
\newcommand{\TD}{{\rm TD}}

\newcommand{\HS}{{\rm HS}}
\newcommand{\NA}{{\rm N}}

\renewcommand{\O}{\emptyset}

\newcommand{\MAX}{{\rm MAX}}
\newcommand{\MAJ}{{\rm MAJ}}
\newcommand{\MIMA}{{\rm MAMI}}
\newcommand{\MAMI}{{\rm MAMI}}
\newcommand{\FF}{{\cal F}}

\newcommand{\DanaF}[2]{{\color{blue} #1 {{\color{mygray} (old: #2)}}}}

\newcommand{\DanaFD}[2]{#1}
\newcommand{\Ray}{{\rm Ray}}
\newcommand{\imply}{{\Rightarrow}}
\newcommand{\De}{{\rm De}}
\newcommand{\As}{{\rm As}}
\newcommand{\Ne}{{\rm Ne}}
\newcommand{\DE}{{\rm DE}}
\newcommand{\AS}{{\rm AS}}
\newcommand{\lca}{{\rm lca}}
\newcommand{\A}{{\cal A}}

\ignore{
\newtheorem{theorem}{Theorem}

\newtheorem{lemma}[theorem]{Lemma}
\newtheorem{definition}[theorem]{Definition}

}

\newlength{\boxwidth}
\newcommand{\singlespacing}%
{\small\normalsize}

\begin{document}

\title{On Polynomial time Constructions of\\ Minimum Height Decision Tree}
\author{Nader H. Bshouty\ \ \ \ Waseem Makhoul}
\institute{Technion, Haifa, Israel\\
bshouty@ca.technion.ac.il}
\maketitle

\begin{abstract}
A decision tree $T$ in $B_m:=\{0,1\}^m$ is a binary tree where each of its internal nodes is labeled with an integer in $[m]=\{1,2,\ldots,m\}$, each leaf is labeled with an assignment $a\in B_m$ and each internal node has two outgoing edges that are labeled with $0$ and $1$, respectively. Let $A\subset \{0,1\}^m$. We say that $T$ is a decision tree for $A$ if (1) For every $a\in A$ there is one leaf of $T$ that is labeled with $a$. (2) For every path from the root to a leaf with internal nodes labeled with $i_1,i_2,\ldots,i_k\in[m]$, a leaf labeled with $a\in A$
and edges labeled with $\xi_{i_1},\ldots,\xi_{i_k}\in\{0,1\}$, $a$ is the only element in $A$ that satisfies $a_{i_j}=\xi_{i_j}$ for all $j=1,\ldots,k$.

Our goal is to write a polynomial time (in $n:=|A|$ and $m$) algorithm that for an input $A\subseteq {B_m}$ outputs a decision tree for $A$ of minimum depth.
This problem has many applications that include, to name a few, computer vision, group testing, exact learning from membership queries and game theory.

Arkin et al. and Moshkov~\citep{AMMRS98,M04} gave
a polynomial time $(\ln |A|)$- approximation
algorithm (for the depth).
The result of Dinur and Steurer~\citep{DS14} for set cover
implies that this problem cannot be approximated with ratio $(1-o(1))\cdot \ln |A|$,
unless P=NP. Moskov studied in~\citep{M04} the combinatorial measure of extended teaching dimension of $A$, $\ETD(A)$. He showed that $\ETD(A)$ is a lower bound for the depth of the decision tree for $A$ and then gave an {\it exponential time} $\ETD(A)/\log(\ETD(A))$-approximation algorithm.

In this paper we further study the $\ETD(A)$ measure
and a new combinatorial measure, $\DEN(A)$, that we call the density of the set $A$. We show that $\DEN(A)$ $\le \ETD(A)+1$. We then give two results. The first result is that the lower bound $\ETD(A)$ of Moshkov for the depth of the decision tree for $A$ is greater than the bounds that are obtained by the classical technique used in the literature. The second result is a polynomial time $(\ln 2)\DEN(A)$-approximation (and therefore $(\ln 2)\ETD(A)$-approximation) algorithm for the depth of the decision tree of $A$.
We also show that a better approximation ratio implies P=NP.

We then apply the above results to learning the class of disjunctions of predicates from membership queries~\citep{BDVY17}. We show that the $\ETD$ of this class is bounded from above by the degree $d$ of its Hasse diagram. We then show that Moshkov algorithm can be run in polynomial time and is $(d/\log d)$-approximation algorithm. This gives optimal algorithms when the degree is constant. For example, learning axis parallel rays over constant dimension space.
\end{abstract}
\section{Introduction}\label{Int}

Consider the following
problem: Given an $n$-element set $A\subseteq B_m:=\{0,1\}^m$ from some class of sets ${\cal A}$
and a hidden element $a\in A$.
Given an oracle that answers queries of the type: ``What is the value of $a_i$?''.
Find a polynomial time algorithm that with an input $A$, asks minimum number
of queries to the oracle and finds the hidden element~$a$.
This is equivalent to constructing a minimum height decision tree
for $A$. A decision tree is a binary tree
where each internal node is labeled with an index from $[m]$
and each leaf is labeled with an assignment $a\in B_m$.
Each internal node has two outgoing edges one that is labeled with~$0$
and the other is labeled with $1$. A node that is labeled with
$i$ corresponds to the query ``Is $a_i=0$?''. An edge that is labeled with $\xi$
corresponds to the answer $\xi$. This decision
tree is an algorithm in an obvious way and its height is the
worst case complexity of the number of queries.
A decision tree $T$ is said to be a {\it decision tree for} $A$
if the algorithm that corresponds to $T$ predicts correctly the hidden assignment $a\in A$.
Our goal is to construct
a small height decision tree for $A\subseteq B_m$ in time polynomial in $m$ and $n:=|A|$.
We will denote by $\OPT(A)$ the minimum height decision tree for $A$.

This problem is related to the following problem in exact learning~\citep{A87}:
Given a class $C$ of boolean functions $f:X\to \{0,1\}$.
Construct in $poly(|C|,|X|)$ time an optimal adaptive algorithm that learns $C$ from membership queries. This learning problem is equivalent to constructing a
minimum height decision tree for the set $A=\{a^{(i)}|a^{(i)}_j=f_i(x_j)\}$
where $f_i$ is the $i$th function in $C$ and $x_j$ is the $j$th instance
in $X$. In computer vision the problem is related to minimizing the
number of ``probes'' (queries) needed to determine which one of a finite
set of geometric figures is present in an image~\citep{AMMRS98}. In game theory
the problem is related to the minimum number of turns required in order to win a guessing game.

\subsection{Previous and New Results}

In \citep{AMMRS98}, Arkin et al. showed that (AMMRS-algorithm) if at every node
the decision tree chooses $i$ that partitions the
current set (the set of assignments that are consistent to the answers of the
queries so far) as evenly as possible,
then the height of the tree is within a factor of $\log |A|$ from optimal. I.e., $\log |A|$-approximation algorithm.
Moshkov~\citep{M04} analysis shows that this algorithm
is $(\ln |A|)$-approximation algorithm.
This algorithm runs in polynomial time in $m$ and $|A|$.

Hyafil and Rivest, \citep{HR76}, show
that the problem of constructing a minimum depth decision tree is NP-Hard.
The reduction of Laber and Nogueira, \citep{LN04} to set cover with
the inapproximability result of Dinur and Steurer~\citep{DS14} for set cover
implies that it cannot be approximated to a factor of $(1-o(1))\cdot \ln |A|$
unless P=NP.  Therefore,
no better approximation ratio can be obtained if no
constraint is added to the set $A$.

Moshkov, \citep{M83}, studied the extended teaching dimension combinatorial measure, $\ETD(A)$, of a set $A\subseteq B_m$. It is the maximum over all the possible assignments $b\in B_m$ of the minimum number of indices $I\subset [m]$ in which $b$ agrees with at most one $a\in A$.
Moshkov showed two results. The first is that $\ETD(A)$ is a lower bound for $\OPT(A)$. The second is an exponential time algorithm that asks $(2\ETD(A)/\log \ETD(A))\log n$ queries. This gives a $(\ln 2)$ $(\ln |A|)/\log \ETD(A)$ -approximation (exponential time) algorithm (since $\OPT(A)\ge \ETD(A)$) and at the same time $2\ETD$ $(A)/\log \ETD(A)$-approximation algorithm (since $\OPT(A)\ge \log |A|$). Since many interesting classes have small $\ETD$ dimension, the latter result gives small approximation ratio but unfortunately Moshkov algorithm runs in exponential time.

In this paper we further study the $\ETD$ measure. We show that any polynomial time $(1-o(1))\ETD(C)$-approximation algorithm implies P=NP. Therefore, Moshkov algorithm cannot run in polynomial time unless P=NP. We then show that the above AMMRS-algorithm, \citep{AMMRS98}, is polynomial time $(\ln 2)\ETD(C)$-approximation algorithm. This gives a small approximation ratio for classes with small extended teaching dimension.

Another reason for studying the \ETD\ of classes is the following: If you find the $\ETD$ of the set $A$ then you either get a lower bound that is better than the information theoretic lower bound $\log |A|$ or you get an approximation algorithm with a better ratio than $\ln |A|$. This is because if $\ETD(A)<\log |A|$ then the AMMRS-algorithm has a ratio $(\ln 2)\ETD(A)$ that is better than the $\ln |A|$ ratio and if $\ETD(A)>\log |A|$ then Moshkov lower bound, $\ETD(A)$, for $\OPT(A)$ is better than the information theoretic lower bound $\log |A|$.

To get the above results, we define a new combinatorial measure called the {\it density} $\DEN(A)$ of the set $A$. If $Q=\DEN(A)$ then there is a subset $B\subseteq A$ such that an adversary can give answers to the queries that eliminate at most $1/Q$ fraction of the number of elements in $B$. This forces the learner to ask at least $Q$ queries. We then show that $\ETD(A)\ge \DEN(A)-1$. On the other hand, we show that if $Q=\DEN(A)$ then a query in the AMMRS-algorithm eliminates at least $(1-1/Q)$ fraction of the assignments in $A$. This gives a polynomial time $(\ln 2)\DEN(A)$-approximation algorithm which is also a $(\ln 2)(\ETD(A)+1)$-approximation algorithm.

In order to compare both algorithms we show that $(\ETD(A)-1)/\ln |A|\le \DEN(A)\le \ETD(A)+1$ and for random uniform $A$ (and therefore for almost all $A$), with high probability $\DEN(A)=\Theta(\ETD(A)/\ln |A|)$. Since $|A|>\ETD(A)$, this shows that AMMRS-algorithm may get a better approximation ratio than Moshkov algorithm.

The inapproximability results follows from the reduction of Laber and Nogueira, \citep{LN04} to set cover with
the inapproximability result of Dinur and Steurer~\citep{DS14} and the fact that $\DEN(A)\le \ETD(A)+1\le\OPT(A)+1$.

We then apply the above results to learning the class of disjunctions of predicates from a set of predicates $\FF$ from membership queries~\citep{BDVY17}. We show that the $\ETD$ of this class is bounded from above by the degree $d$ of its Hasse diagram. We then show that Moshkov algorithm, for this class, runs in {\it polynomial time} and is $(d/\log d)$-approximation algorithm. Since $|\FF|\ge d$ (and in many applications, $|\FF|\gg d$), this improves the $|\FF|$-approximation algorithm SPEX in~\citep{BDVY17} when the size of Hasse diagram is polynomial. This also gives optimal algorithms when the degree $d$ is constant. For example, learning axis parallel rays over constant dimension space.

\section{Definitions and Preliminary Results}\label{Se2}
In this section we give some definitions and preliminary results

\subsection{Notation}
Let $B_m=\{0,1\}^m$. Let $A=\{a^{(1)},\ldots,a^{(n)}\}\subseteq B_m$ be an $n$-element set. We will write $|A|$ for the number of elements in $A$. For $h\in B_m$ we define $A+ h=\{a+h|a\in A\}$ where $+$ (in the square brackets) is the bitwise exclusive or of elements in $B_m$. \ignore{We define $z(A)=1$ if $A$ contains a zero vector and $z(A)=0$ otherwise.}

For integer $q$ let $[q]=\{1,2,\ldots,q\}$.
Throughout the paper, $\log x=\log_2x$.

\subsection{Optimal Algorithm}
\ignore{A {\it decision tree} $T$ in $B_m$ is a binary tree where each internal nodes
is labeled with an integer in $[m]$, each leaf is labeled with $a\in B_m$ and each internal node has two outgoing edges that are labeled with $0$ and $1$, respectively. We say that $T$ is a decision tree for $A$ if every $a\in A$ is in some leaf of $T$ and for every path from the root to a leaf with internal nodes labeled with $i_1,i_2,\ldots,i_k\in[m]$, a leaf labeled with $a\in A$
and edges labeled with $\xi_{i_1},\ldots,\xi_{i_k}\in\{0,1\}$, $a$ is the only assignment in $A$ that satisfies $a_{i_r}=\xi_{i_r}$ for all $r=1,\ldots,k$.}

We denote by $\OPT(A)$ the minimum depth of a decision tree for $A$.
Our goal is to build a decision tree for $A$ with small depth.

\ignore{When, for each level $\ell$ of the decision tree $T$, all the nodes at that level contain the same index $i_\ell\in [m]$ then we call the tree a {\it non-adaptive decision tree} for $A$. The minimum depth of a non-adaptive decision tree for $A$ is denoted by $\OPT_{\NA}(A)$. Obviously, each $j\in [n]$ occurs in exactly one leaf and therefore}
Obviously
\begin{eqnarray}\label{logli}
\log n\le \OPT(A)\le n-1
\end{eqnarray}
where $n:=|A|$.

\ignore{The above problem is equivalent to the following problem. Give a $A\subseteq B_m$. Two players P1 and P2 know the set $A$.
Player P1 hides an element $a\in A$ and player P2 goal is to find this element with minimum number of questions of the form ``What is the $i$th entry of the hidden $a$?''.
Then $\OPT(A)$ is the minimum number
of questions that is needed to find the hidden column
with unlimited computational power.

In learning theory, $\OPT(A)$ is the minimum number of membership queries that is needed to learn the class of boolean functions $C_A:=\{f_a:[m]\to \{0,1\}\ |\ a\in A, f_a(i)=a_{i}\}$ with an adaptive algorithm. }

The following result is easy to prove (see Appendix A)
\begin{lemma}\label{optaph}  We have
$\OPT(A)=\OPT(A+h).$
\end{lemma}

\subsection{Extended Teaching Dimension}
In this section we define the extended teaching dimension.

Let $h\in B_m$ be any element.
We say that
a set $S\subseteq [m]$ is a {\it specifying set for $h$ with respect to $A$}
if $|\{a\in A\ |\ (\forall i\in S) h_i=a_i\}|\le 1$. That is, there is
at most one element in $A$ that is {\it consistent with} $h$ on the entries of~$S$. Denote by $\ETD(A,h)$
the minimum size of a specifying set for $h$ with respect to $A$.
The {\it extended teaching dimension} of $A$ is
\begin{eqnarray}\label{defETD}
\ETD(A)=\max_{h\in B_m} \ETD(A,h).
\end{eqnarray}
We will write $\ETD z(A)$ for $\ETD(A,0)$. It is easy to see that
\begin{eqnarray}\label{sah}
\ETD(A,h)=\ETD z(A+h)\mbox{\ and\ }\ETD(A)=\ETD(A+h).
\end{eqnarray}

We say that a set $S\subseteq [m]$ is a {\it strong specifying set for $h$ with respect to $A$}
if either $h\in A$ and $|\{a\in A\ |\ (\forall i\in S) h_i=a_i\}|= 1$, or $|\{a\in A\ |\ (\forall i\in S) h_i=a_i\}|= 0$. That is, if $h\in A$ then there is
exactly one element in $A$ that is {\it consistent with} $h$ on the entries of~$S$. Otherwise, no element in $A$
is consistent with $h$ on $S$. Denote $\SETD(A,h)$
the minimum size of a strong specifying set for $h$ with respect to $A$.
The {\it strong extended teaching dimension} of $A$ is
\begin{eqnarray}\label{setd}
\SETD(A)=\max_{h\in B_m} \SETD(A,h).
\end{eqnarray}
We will write $\SETD z(A)$ for $\SETD(A,0)$.
It is easy to see that
\begin{eqnarray}\label{ssah}
\SETD(A,h)=\SETD z(A+h)\mbox{\ and \ }\SETD(A)=\SETD(A+h).
\end{eqnarray}
Obviously, $\ETD(A,h)\le \min(m,n-1)\ \mbox{and}\ \ETD(A,h)\le \SETD(A,h)\le \min(m,n)$
\ignore{\begin{eqnarray}\label{snss}

\end{eqnarray}}

We now show
\begin{lemma} \label{SSS} We have
$\ETD(A,h)\le \SETD(A,h)\le \ETD(A,h)+1$ and therefore
$\ETD(A)\le \SETD(A)\le \ETD(A)+1.$
\end{lemma}
\begin{proof}
The fact $\ETD(A,h)\le \SETD(A,h)$ follows from the definitions.
Let $S\subseteq [m]$ be a specifying set for $h$ with respect to $A$. Then for
$T:=\{a\in A\ |\ (\forall i\in S) h_i=a_i\}$ we have $t:=|T|\le 1$. If $t=0$ or $h\in A$ then $S$ is a strong specifying set for $h$ with respect to $A$. If $t=1$ and $h\not \in A$ then for the element $a\in T$ there is $j\in [m]$ such that $a_j\not=h_j$ and then $S\cup\{j\}$ is a strong specifying set for $h$ with respect to $A$. This proves that $\SETD(A,h)\le \ETD(A,h)+1$.

The other claims follows immediately.\qed
\end{proof}

Obviously, for any $B\subseteq A$
\begin{eqnarray}\label{EBA}
\ETD(B)\le \ETD(A),\ \ \ \SETD(B)\le \SETD(A).
\end{eqnarray}

\subsection{Hitting Set}
In this section
we define the hitting set for $A$.

A {\it hitting set for $A$} is a set $S\subseteq [m]$ such that for every non-zero element $a\in A$ there is $j\in S$ such that $a_j=1$. That is, $S$ {\it hits} every element in $A$ except the zero element (if it exists). The size of the minimum size hitting set for $A$ is denoted by $\OPTHS(A)$.

We now show
\begin{lemma} \label{HSSS} We have
$\HS(A)= \SETD z(A).$
In particular, $\SETD(A,h)=\HS(A+ h)$
and
$\SETD(A)=\max_{h\in B_m} \HS(A+ h).$
\end{lemma}
\begin{proof}
If $0\in A$ then $\SETD z(A)$ is the minimum size of a set $S$
such that $\{a\in A\ |\ (\forall i\in S) a_i=0\}=\{0\}$ and if $0\not\in A$ then it is the minimum size of a set $S$ such that $\{a\in A\ |\ (\forall i\in S) a_i=0\}=\O$. Therefore the set $S$ hits all the nonzero elements in $A$.

The other results follow from (\ref{ssah}) and the definition of $\SETD$.\qed
\end{proof}

\subsection{Density of a Set}\label{DM}
In this section we define our new measure $\DEN$ of a set.

Let $A=\{a^{(1)},\ldots,a^{(n)}\}\subseteq B_m$. We define $\MAJ(A)\in B_m$ such that $\MAJ(A)_i=1$ if the number of ones in $(a^{(1)}_i,\cdots,a^{(n)}_i)$ is greater or equal the number of zeros and $\MAJ(A)_i=0$ otherwise.  We denote by $\MAX(A)$ the maximum number of ones in $(a^{(1)}_i,\cdots,a^{(n)}_i)$ over all $i=1,\ldots,m$. Let
\begin{eqnarray}\label{mima}
\MIMA(A)=\min_{h\in B_m} \MAX(A+ h)=\MAX(A+ \MAJ(A)).
\end{eqnarray}
For $j\in [m]$ and $\xi\in\{0,1\}$ let $A_{j,\xi}=\{a\in A\ |\ a_j=\xi\}$. Then
\begin{eqnarray}\label{mima2}
\MIMA(A)=\max_j\min(|A_{j,0}|,|A_{j,1}|).
\end{eqnarray}

We define the {\it density} of a set $A\subseteq B_m$ by

\begin{eqnarray}\label{den}
\DEN(A)= \max_{B\subseteq A} \frac{|B|-1}{\MIMA(B)}.
\end{eqnarray}

Notice that since every $j\in [m]$ can hit at most $\MAX(A)$ elements in $A$ we have
\begin{eqnarray}\label{HIT}
\HS(A)\ge \frac{|A|-1}{\MAX(A)}.
\end{eqnarray}

\section{Bounds for $\OPT$}
In this section we give upper and lower bounds for $\OPT$.

\subsection{Lower Bound}
Moshkov results in~\citep{M83,H95} and the information theoretic bound in (\ref{logli}) give the following lower bound. We give the proof in Appendix A for completeness.
\begin{lemma}~\citep{M83,H95} \label{LBo1} Let $A\subseteq B_m$ be any set. Then
$$\OPT(A)\ge \max(\ETD(A),\log |A|).$$
\end{lemma}

Many lower bounds in the literature for $\OPT(A)$ are based on finding a subset $B\subseteq A$
such that for each query there is an answer that
eliminates at most small fraction $E$ of $B$. Then $(|B|-1)/E$ is a lower bound for $\OPT(A)$.
The best possible bound that one can get using this technique is exactly $\DEN(A)$ (Lemma~\ref{FLB}),
the density defined in Section~\ref{DM}.
Lemma~\ref{SLB} shows that the lower bound $\ETD(A)$ for $\OPT(A)$ exceeds any such bound.

In Appendix A we prove
\begin{lemma}\label{FLB}
We have
$\OPT(A)\ge \DEN(A).$
\end{lemma}

\begin{lemma}\label{SLB} We have
$\ETD(A)\ge \DEN(A)-1.$
\end{lemma}
\begin{proof} By (\ref{mima}) and (\ref{den}) there is $B\subseteq A$ such that
\begin{eqnarray}\label{jsr}
\DEN(A)=\frac{|B|-1}{\MIMA(B)}=\frac{|B|-1}{\MAX(B+h)}
\end{eqnarray}
where $h=\MAJ(B)$. Then
\begin{eqnarray*}
\ETD(A)&\stackrel{(\ref{EBA})}{\ge}& \ETD(B)\stackrel{(\ref{defETD})}{\ge} \ETD(B,h) \\
&\stackrel{L\ref{SSS}}{\ge}& \SETD(B,h)-1\stackrel{L\ref{HSSS}}{=}\HS(B+ h)-1
\\
&\stackrel{(\ref{HIT})}{\ge}& \frac{|B|-1}{\MAX(B+ h)}-1\stackrel{(\ref{jsr})}{=}\DEN(A)-1. \qed
\end{eqnarray*}
\end{proof}

In Appendix A we also prove
\begin{lemma}~\label{zzz} We have
$$\ETD(A)\le \ln |A|\cdot\DEN(A)+1.$$
\end{lemma}
It is also easy to see (by standard analysis using Chernoff Bound) that for a random uniform $A$, with positive probability, $\DEN(A)=O(1)$ and $\ETD(A)=\Theta(\log |A|)$. See the proof sketch in Appendix A. So the bound in Lemma~\ref{zzz} is asymptotically best possible.

\subsection{Upper Bounds}
Moshkov~\citep{M83,H95} proved the following upper bound. We gave the proof in the Appendix B for completeness.

\begin{lemma}~\citep{M83,H95} \label{LBo2} Let $A\subseteq \{0,1\}^m$ of size $n$. Then
$$\OPT(A)\le \ETD(A)+\frac{\ETD(A)}{\log \ETD(A)}\log n\le \frac{2\cdot \ETD(A)}{\log \ETD(A)}\log n.$$
\end{lemma}

In \citep{M83,H95}, Moshkov gave an example of a $n$-set $A_{E}\subseteq\{0,1\}^m$ with
$\ETD(A_{E})=E$ and $\OPT(A_{E})=\Omega((E/\log E)\log n)$. So the upper bound in the above lemma is the best possible.

\section{Polynomial Time Approximation Algorithm}

Given a a set $A\subseteq B_m$.
Can one construct an algorithm that finds
a hidden $a\in A$ with $\OPT(A)$ queries?
Obviously, with unlimited computational power this can be done
so the question is:
How close to $\OPT(A)$ can one get
when polynomial time $poly(m,n)$ is allowed for the construction?

An exponential time algorithm follows
from the following
$$\OPT(A)=\min_{i\in [m]} \max(\OPT(A_{i,0}),\OPT(A_{i,1}))$$
where $A_{i,\xi}=\{a\in A\ | \ a_i=\xi\}$. This algorithm runs in time at least $m!\ge (m/e)^m$.
See also~\citep{G72,AGMMPS93}.

Can one give a better exponential time algorithm? In what follows (Theorem~\ref{main}) we use  Moshkov~\citep{M83,H95} result (Lemma~\ref{LBo2}) to give a better exponential time approximation algorithm. In Appendix B we give another simple proof of the Moshkov~\citep{M83,H95} result that in practice uses less number of specifying sets. When the extended teaching dimension is constant, the algorithm is $O(1)$-approximation algorithm and runs in polynomial time.

\begin{theorem}\label{main} Let ${\cal A}$ be a class of sets $A\subseteq B_m$ of size $n$.
If there is an algorithm that for any $h\in B_m$ and any $A\in {\cal A}$
gives a specifying set for $h$ with respect to $A$
of size at most~$E$ in time $T$
then there is an algorithm that for any $A\in {\cal A}$ constructs a decision tree for $A$ of depth at most
$$E+\frac{E}{\log E}\log n\le E+\frac{E}{\log E}\OPT(A)$$ queries and runs in time $O(T\log n+ nm)$.
\end{theorem}
\begin{proof} Follows immediately from Moshkov algorithm~\citep{M83,H95}. See Appendix~B.\qed
\end{proof}

The following result immediately follows from Theorem~\ref{main}.

\begin{theorem}\label{T1} Let $A\subseteq B_m$ be a $n$-set. There is an
algorithm that finds the hidden column in time
$${m\choose \ETD(A)}\cdot\ETD(A)\cdot n\log n$$ and asks at most
$$\frac{2\cdot \ETD(A)\cdot\log n}{\log \ETD(A)}
\le \frac{2\cdot \min(\ETD(A),\log n)}{\log \ETD(A)}\OPT(A)$$ queries.

In particular, if $\ETD(A)$ is constant then
the algorithm is $O(1)$-approximation algorithm that
runs in polynomial time.
\end{theorem}
\begin{proof} To find a specifying set for $h$ with respect to $A$ we exhaustively check each $\ETD(A)$ row of $A$. Each check takes time $n$. Since the algorithm asks at most $\ETD(A)\cdot \log n$ queries, the time complexity is as stated in the Theorem.
\end{proof}

Can one do it in $poly(m,n)$ time?
Hyafil and Rivest, \citep{HR76}, show that the problem of finding $\OPT$ is NP-Complete.
The reduction of Laber and Nogueira, \citep{LN04}, of set cover to this problem with
the inapproximability result of Dinur and Steurer~\citep{DS14} for set cover
implies that it cannot be approximated to $(1-o(1))\cdot \ln n$
unless P=NP.

In \citep{AMMRS98}, Arkin et al. showed that (the AMMRS-algorithm) if at the $i$th query the algorithm chooses
an index $j$ that partitions the current node set (the elements in $A$ that are consistent with the answers until this node)
$A$ as evenly as possible, that is, that maximizes
$\min(|\{a\in A|a_j=0\}|,|\{a\in A|a_j=1\}|)$,
then the query complexity is within a factor of $\lceil \log n \rceil$ from optimal. The
AMMRS-algorithm, \citep{AMMRS98}, runs in time $poly(m,n)$.
Moshkov~\citep{AMMRS98,M04} analysis shows that this algorithm
is $\ln n$-approximation algorithm and therefore is optimal.
In this section we will give a simple proof.

In \citep{M83,H95}, Moshkov gave a simple $\ETD(A)$-approximation algorithm (Algorithm MEMB-HALVING-1 in~\citep{H95}). He then gave another algorithm that achieves the query complexity in Lemma~\ref{LBo2} (Algorithm MEMB-HALVING-2 in~\citep{H95}). This is within a factor of
$$\frac{2\cdot \min(\ETD(A),\log n)}{\log \ETD(A)}$$ from optimal.
This is better than the ratio $\ln n$, but, unfortunately,
both algorithms require finding a minimum size specifying
set and the problem of finding a minimum size specifying set for $h$ is
NP-Hard, \citep{SM91,ABCS92,GK95}.

Can one achieve a $O(\ETD(A))$-approximation. In the following we give a surprising result. We show that the AMMRS-algorithm is $(\ln 2) \ETD(A)$-approximation algorithm. We also show that no better ratio can be achieved unless P=NP.

\begin{theorem}\label{T2} The AMMRS-algorithm runs in time $O(mn)$
and finds the hidden element $a\in A$ with at most
\begin{eqnarray*}
\DEN(A)\cdot \ln(n) &\le& \min((\ln 2)\DEN(A),\ln n)\cdot \OPT(A)\\
&\le& \min((\ln 2)(\ETD(A)+1),\ln n )\cdot \OPT(A)
\end{eqnarray*}
queries.
\end{theorem}
\begin{proof} Let $B$ be any subset of $A$. Then,
$$\DEN(B)\stackrel{(\ref{den})}{\ge}\frac{|B|-1}{\MIMA(B)}$$ and therefore
$$\MIMA(B)\ge \frac{|B|-1}{\DEN(B)}  \ge \frac{|B|-1}{\DEN(A)}.$$

Since the AMMRS-algorithm chooses at each node in the decision tree the index $j$ that maximizes $\min(|B_{j,0}|,|B_{j,1}|)$ where $B_{j,\xi}=\{a\in B|a_j=\xi\}$ and $B$ is the set of elements in $A$ that are consistent with the answers until this node, we have
\begin{eqnarray*}
\max(|B_{j,0}|,|B_{j,1}|)-1&=& |B|-1-\min(|B_{j,0}|,|B_{j,1}|)\\
&\stackrel{(\ref{mima2})}{=}& |B|-1-\MIMA(B)\le (|B|-1)\left(1-\frac{1}{\DEN(A)}\right).
\end{eqnarray*}
Therefore, for a node $v$ of depth $h$ in the decision tree, the set $B(v)$ of elements in $A$ that are consistent with the answers until this node contains at most
$$(|A|-1)\left(1-\frac{1}{\DEN(A)}\right)^h+1$$ elements. Therefore the depth of the tree is at most
$$\DEN(A)\ln |A|.\qed$$
\end{proof}

We now show that our approximation algorithm is optimal

\begin{theorem} Let $\epsilon$ be any constant. There is no polynomial time algorithm that finds the hidden element with less than $(1-\epsilon)\DEN(A)\cdot \ln|A|$ unless P=NP.
\end{theorem}
\begin{proof} Suppose such an algorithm exists. Then
$$(1-\epsilon)\DEN(A)\ln |A|\stackrel{L\ref{FLB}}\le (1-\epsilon)\ln |A| \OPT(A).$$
That is, the algorithm is also $(1-\epsilon)\ln |A|$-approximation algorithm.
Laber and Nogueira, \citep{LN04} gave a polynomial time algorithm reduction of minimum depth decision tree to set cover and Dinur and Steurer~\citep{DS14} show that there is no polynomial time $(1-o(1))\cdot \ln |A|$ for set cover unless P=NP.
Therefore, such an algorithm implies P=NP.\qed
\end{proof}

\section{Applications to Disjunction of Predicates}
In this section we apply the above results to learning the class of disjunctions of predicates from a set of predicates $\FF$ from membership queries~\citep{BDVY17}.

Let $C=\{f_1,\ldots,f_n\}$ be a set of boolean functions $f_i:X\to \{0,1\}$ where $X=\{x_1,\ldots,x_m\}$. Let $A_C=\{(f_i(x_1),\ldots,f_i(x_m))\ |\ i=1,\ldots,n\}$.
We will write $\OPT(A_C), \ETD(A_C),$ etc. as $\OPT(C), \ETD(C),$ etc.

\DanaFD{Let $\FF$ be}{Given} a set of boolean functions (predicates) over a domain $X$. We consider the class of functions $\FF_\vee:=\{\vee_{f\in S}f\ |\ S\subseteq \FF\}$.
\subsection{\DanaFD{An Equivalence Relation Over $\FF_\vee$}{Partial Order \DanaF{over}of $\FF_\vee$}}\label{sec:defs}
In this section, we present an equivalence relation over $\FF_\vee$ and define the representatives of the equivalence classes. This enables us in later sections to focus on the representative elements from $\FF_\vee$.
Let $\FF$ be a set of boolean functions over the domain $X$. \DanaFD{The}{Define the} equivalence relation $=$ over $\FF_\vee$ \DanaFD{is defined as follows:}{where} two disjunctions $F_1,F_2\in \FF_\vee$ are equivalent ($F_1= F_2$) if $F_1$ is logically equal to $F_2$. In other words, they represent the same function \DanaFD{(from $X$ to $\{0,1\}$)}{$X\to \{0,1\}$}. We \DanaFD{}{will} write $F_1\equiv F_2$ to denote that $F_1$ and $F_2$ are identical\DanaFD{; that}{. That} is, they have the same representation. For example, consider $f_1,f_2:\{0,1\}\to \{0,1\}$ where $f_1(x)=1$ and $f_2(x)=x$. Then, $f_1\vee f_2 = f_1$ but $f_1 \vee f_2 \not\equiv f_1$.

We denote by $\FF_\vee^*$ the set of equivalence classes of $=$ and write each equivalence class as $[F]$, where $F\in\FF_\vee$. Notice that if $[F_1]=[F_2]$, then $[F_1\vee F_2]=[F_1]=[F_2]$. Therefore, for every $[F]$, we can choose the {\it representative element} to be $G_F:=\vee_{F'\in S}F'$ where $S\subseteq \FF$ is the maximum size set that satisfies $\vee S:=\vee_{f\in S}f=F$. We denote by $G(\FF_\vee)$ the set of all representative elements. Accordingly, $G(\FF_\vee)=\{G_F\ |\ F\in\FF_\vee\}$.
\DanaFD{As an example, consider}{Before we proceed we give an example. Consider} the set $\FF$ consisting of four functions $f_{11},f_{12},f_{21},f_{22}:\{1,2\}^2\to \{0,1\}$ where $f_{ij}(x_1,x_2)=[x_i\ge j]$ where $[x_i\ge j]=1$ if $x_i\ge j$ and $0$ otherwise. There are $2^4=16$ elements in $\Ray^2_2:=\FF_\vee$ and five representative functions in $G(\FF_\vee)$: $G(\FF_\vee)=\{f_{11}\vee f_{12}\vee f_{21}\vee f_{22}$, $f_{12}\vee f_{22}$, $f_{12}$, $f_{22},0\}$ (where $0$ is the zero function).

\subsection{A Partial Order Over $\FF_\vee$ and Hasse Diagram}\label{sec:32}
In this section, we define a partial order over $\FF_\vee$ and present related definitions.
The partial order, denoted by $\imply$, is defined as follows:
$F_1\imply F_2$ if $F_1$ logically implies~$F_2$.
Consider the Hasse diagram $H(\FF_\vee)$ of $G(\FF_\vee)$ for this partial order. The maximum (top) element in the diagram is $G_{\max}:=~\vee_{f\in \FF}f$. The minimum (bottom) element is $G_{\min}:=\vee_{f\in \O}f$, i.e., the zero function. Figures~\ref{HasseClause32} and \ref{RAY23E} shows an illustration of the Hasse diagram.

In a Hasse diagram, $G_1$ is a {\it descendant} (resp., {\it ascendent}) of $G_2$ if there is a (nonempty) downward path from $G_2$ to $G_1$ (resp., from $G_1$ to $G_2$), i.e., $G_1\imply G_2$ (resp., $G_2\imply G_1$) and $G_1\not=G_2$. $G_1$ is an {\it immediate descendant} of $G_2$ in $H(\FF_\vee)$ if $G_1\imply G_2$, $G_1\not=G_2$ and there is no $G\in G(\FF_\vee)$ such that $G\not= G_1$, $G\not=G_2$ and $G_1\imply G\imply G_2$. $G_1$ is an {\it immediate ascendant} of $G_2$ if $G_2$ is an immediate descendant of $G_1$.

We denote by $\De(G)$ and $\As(G)$ the sets of all the immediate descendants and immediate ascendants of \DanaFD{}{and neighbours of}$G$, respectively. The {\it neighbours set} of $G$ is
$\Ne(G)=\De(G)\cup \As(G)$. We further denote by $\DE(G)$ and $\AS(G)$ the sets of all $G$'s descendants and ascendants, respectively.
\begin{definition} The {\it degree} of $G$ is $\deg(G)=|\Ne(G)|$ and the degree $\deg(\FF_\vee)$ of $\FF_\vee$ is $\max_{G\in G(\FF_\vee)}\deg (G)$.
\end{definition}

For $G_1$ and $G_2$, we define their {\it lowest common ascendent} (resp., greatest common descendant) $G=\lca(G_1,G_2)$ (resp., $G=\gcd(G_1,G_2)$) to be the minimum (resp., maximum) element in $\AS(G_1)\cap \AS(G_2)$ (resp., $\DE(G_1)\cap \DE(G_2)$).
\DanaFD{}{Therefore, we can show:}

The following result is from~\citep{BDVY17}
\begin{lemma}\label{lca} Let $G_1,G_2\in G(\FF_\vee)$. Then, $\lca(G_1,G_2)=G_1\vee G_2$.

In particular, if $G_1,G_2$ are two distinct immediate descendants of $G$, then $G_1\vee G_2=G$.
\end{lemma}

\subsection{Witnesses}
In this subsection we define the term \emph{witness}.
Let $G_1$ and $G_2$ be elements in $G(\FF_\vee)$. An element \DanaFD{$a\in X$}{$a$} is a {\it witness} for $G_1$ and $G_2$ if $G_1(a)\not= G_2(a)$.

For a class of boolean functions $C$ over a domain $X$ and a function $G\in C$ we say that a set of elements $W\subseteq X$ is a {\it witness set} for $G$ in $C$ if for every $G'\in C$ and $G'\not=G$ there is a witness in $W$ for $G$ and $G'$.

\subsection{The Extended Teaching Dimension of $\FF_\vee$}
In this section we prove
\begin{lemma}\label{etdf} For every $h:X\to \{0,1\}$ if $h\nRightarrow G_{\max}$ then $\ETD(\FF_\vee,h)=1$. Otherwise, there is $G\in G(\FF_\vee)$ such that
$$\ETD(\FF_\vee,h)\le |\De(G)|+\HS(\As(G)\wedge \bar{G})\le |\Ne(G)|=\deg(G)$$
where $As(G)\wedge \bar{G}=\{s\wedge \bar{G}\ |\ s\in \As(G)\}$.
In particular,
$$\ETD(\FF_\vee)\le \max_{G\in G(\FF_\vee)} \left(|\De(G)|+\HS(\As(G)\wedge \bar{G})\right)\le \deg(\FF_\vee).$$
\end{lemma}
\begin{proof} Let $h:X\to \{0,1\}$ be any function. If $h\nRightarrow G_{\max}$ then there is an assignment $a$ that satisfies $h(a)=1$ and $G_{\max}(a)=0$. Since for all $G\in G(\FF_\vee)$, $G\Rightarrow G_{\max}$ we have $G(a)=0$. Therefore, the set $\{a\}$ is a specifying set for $h$ with respect to $\FF_\vee$ and $\ETD(\FF_\vee,h)=1$.

Let $h\Rightarrow G_{\max}$. Consider any $G\in G(\FF_\vee)$ such that $h\imply G$ and for every immediate descendant $G'$ of $G$ we have $h\nRightarrow G'$. Now for every immediate descendent $G'$ of $G$ find an assignment $a$ such that $G'(a)=0$ and $h(a)=1$. Then $a$ is a witness for $h$ and $G'$. Therefore, $a$ is also a witness for $h$ and every descendant of $G'$. Let $A$ be the set of all such assignments, i.e., for every descendant of $G$ one witness. Then $|A|\le |\De(G)|$ and $A$ is a witness set for $h$ and all the descendants of $G$.
We note here that if
$h=0$ then $G=G_{\min}$ which has no immediate descendants and then $A=\O$.

Consider a hitting set $B$ for $As(G)\wedge \bar{G}$ of size $\HS(As(G)\wedge \bar{G})$.
Now for every immediate ascendant $G''$ of $G$ find an assignment $b\in B$ such that $G''(b)\wedge \bar{G}(b)=1$. Then $G''(b)=1$ and $G(b)=0$. Since $G(b)=0$ we have $h(b)=0$ and then $b$ is a witness for $h$ and $G''$. Therefore, $b$ is also a witness for $h$ and every ascendant of $G''$.
Thus $B$ is a witness set for $h$ in all the ascendants of $G$.

Let $G_0$ be any element in $G(\FF_\vee)$ (that is not a descendant or an ascendant). Consider $G_1=\lca(G,G_0)$. By Lemma~\ref{lca}, we have $G_1=G\vee G_0$. Since $G_1$ is an ascendent of $G$ there is a witness $a\in B$ such that $G_1(a)=1$ and $G(a)=0$. Then $G_0(a)=1$, $h(a)=0$ and $a$ is a witness of $h$ and $G_0$. Therefore $A\cup B$ is a specifying set for $h$ with respect to $G(\FF_\vee)$. Since for every $F\in\FF_\vee$ we have $F=G_F\in G(\FF_\vee)$, $A\cup B$ is also a specifying set for $h$ with respect to $\FF_\vee$.

Since
$$\ETD(\FF_\vee,h)\le |A|+|B|\le |\De(G)|+\HS(\As(G)\wedge \bar{G})$$
the result follows.\qed
\end{proof}
In Appendix C we show that $$\ETD(\FF_\vee)=\max_{G\in G(\FF_\vee)} \left(|\De(G)|+\HS(\As(G)\wedge \bar{G})\right).$$
We could have replaced $|\De(G)|$ by $\HS(\overline{\De(G)}\wedge G)$, but Lemma~\ref{llll} in Appendix C shows that they are both equal.

The following result follows immediately from the proof of Lemma~\ref{etdf}
\begin{lemma} For any $h:X\to \{0,1\}$, a specifying set for $h$ with respect to $\FF_\vee$ of size $\deg(\FF_\vee)$ can be found in time $O(nm)$.
\end{lemma}

By Theorem~\ref{main} we have
\begin{theorem}
There is an algorithm that learns $\FF_\vee$ in time $O(nm)$ and asks at most
$$\deg(\FF_\vee)+\frac{\deg(\FF_\vee)}{\log \deg(\FF_\vee)}\log n\le \left(\frac{\deg(\FF_\vee)}{\log \deg(\FF_\vee)}+1\right) \OPT(\FF_\vee)$$ membership queries.
\end{theorem}

\subsection{Learning Other Classes}

If a specifying set of small size cannot be found in polynomial time then from Theorem~\ref{T1}, \ref{T2} and Lemma~\ref{etdf}, we have
\begin{theorem}\label{mTH} For a class $C$ we have
\begin{enumerate}
\item There is an algorithm that learns $C$ in time $${m\choose \deg(C)}\cdot\ETD(C)\cdot n\log n$$ and asks at most
$$\frac{2\cdot \ETD(C)\cdot\log n}{\log \ETD(C))}
\le \frac{2\cdot \min(\ETD(C)),\log n)}{\log \ETD(C))}\OPT(C)$$ membership queries.

In particular, when $\ETD(C)$ is constant the algorithm runs in polynomial time and its query complexity is (asymptotically) optimal.

\item There is an algorithm that learns $C$ in time $O(nm)$ and asks at most
\begin{eqnarray*}
\DEN(C)\cdot \ln(n) &\le& \min((\ln 2)\DEN(C),\ln n)\cdot \OPT(C)\\
&\le& \min((\ln 2)(\ETD(C)+1),\ln n)\cdot \OPT(C)
\end{eqnarray*} membership queries.

\end{enumerate}
\end{theorem}

\bibliography{bib}
\bibliographystyle{plain}

\newpage

\section{Appendix A}
In this Appendix we give a proof of some lemmas

\noindent{{\bf Lemma~\ref{optaph}}}.
We have
$$\OPT(A)=\OPT(A+h).$$

\begin{proof} Since $(A+h)+h=A$, it is enough to prove that $\OPT(A+h)\le \OPT(A)$.
Now given a decision tree $T$ for $A$ of depth $\OPT(A)$. For each node, $v$, in $T$ labeled with $j$, such that $h_j=1$, exchange the labels in their outgoing edges. Then change the label of each leaf labeled with $a$ to $a+h$. It is easy to show that the new tree is a decision tree for $A+h$.\qed
\end{proof}

\noindent{{\bf Lemma~\ref{LBo1}}}.~\citep{M83,H95} Let $A\subseteq \{0,1\}^m$ be any set. Then
$$\OPT(A)\ge \max(\ETD(A),\log |A|).$$

\begin{proof} The lower bound $\log |A|$ is the information theoretic lower bound. We now prove the other bound.

Let $T$ be a decision tree for $A=\{a^{(1)},\ldots,a^{(n)}\}$ of minimum depth.
Consider the path $P$ in $T$ that at each level chooses the edge that is labeled with $0$.
Let $S$ be the set of labels in the internal nodes of $P$ and $a^{(j)}$ be the label of the leaf of $P$.
Then $a^{(j)}$ is the only element in $A$ that satisfies $a^{(j)}_i=0$ for all $i\in S$.
Therefore $S$ is a specifying set for $0$ with respect to $A$. Thus $\OPT(A)\ge |S|\ge \ETD z(A)$.
Now, by Lemma~\ref{optaph}, for any $h\in \{0,1\}^m$ we have $\OPT(A)=\OPT(A+h)\ge\ETD z(A+h)=\ETD(A,h)$ and therefore $\OPT(A)\ge \max_h \ETD(A,h)=\ETD(A).$ \qed
\end{proof}

\noindent{{\bf Lemma~\ref{FLB}}}.
We have
$\OPT(A)\ge \DEN(A).$

\begin{proof} Let $B\subseteq A$ be a set such that
$$\DEN(A)\stackrel{(\ref{den})}{=} \frac{|B|-1}{\MIMA(B)}\stackrel{(\ref{mima})}{=}\frac{|B|-1}{\MAX(B+\MAJ(B))}.$$
For every query $i\in [m]$ (what is ``$a_i$''?), the adversary answers $\MAJ(B)_i$. This eliminates at most $\MAX(B+\MAJ(B))$ elements from $B$. Therefore the algorithm is forced to ask at least $({|B|-1})/{\MAX(B+\MAJ(B))}$ queries.\qed
\end{proof}

\noindent{\bf Lemma~\ref{zzz}} We have
$$\ETD(A)\le \ln |A|\cdot\DEN(A)+1.$$

\begin{proof}
There is $h_0\in \{0,1\}^m$ such that
\begin{eqnarray}\label{tempooo}
\ETD(A) \stackrel{L\ref{SSS}}{\le}\SETD(A)\stackrel{(\ref{setd})}{=} \SETD(A,h_0)\stackrel{L\ref{HSSS}}{=}\HS(A+h_0).
\end{eqnarray}
For any $C\subseteq A$ we have
\begin{eqnarray*}
\DEN(C) &\stackrel{(\ref{den})}{=}& \max_{B\subseteq C} \frac{|B|-1}{\MIMA(B)}\\
&\stackrel{(\ref{mima})}{\ge}& \max_{B\subseteq C} \frac{|B|-1}{\MAX(B+h_0)}\\
&{\ge}&  \frac{|C|-1}{\MAX(C+h_0)}\\
\end{eqnarray*}
and therefore, for any $C\subseteq A$ we have
\begin{eqnarray}\label{pino}
\MAX(C+h_0)&\ge& \frac{|C|-1}{\DEN(C)}\stackrel{(\ref{den})}{\ge}\frac{|C+h_0|-1}{\DEN(A)}.
\end{eqnarray}
We now consider the following sequence of subsets of $A+h_0$, $C_0,C_1,\ldots,C_t$
where $C_0=A+h_0$ and
the subset $C_{i+1}$ is defined by $C_i$ as follows: Since (\ref{pino}) is also true for $C_i$ there is $j_i\in [m]$ such that $j_i$ hits at least $(|C_i|-1)/\DEN(A)$ elements in $C_i$. Then $C_{i+1}$ contains all the elements in $C_i$ that are not hit by $j_i$. Then
$$|C_{i+1}|-1\le |C_i|-\frac{|C_i|-1}{\DEN(A)}-1=(|C_i|-1)\left(1-\frac{1}{\DEN(A)}\right).$$
Therefore
$$|C_i|\le (|A|-1)\left(1-\frac{1}{\DEN(A)}\right)^i+1.$$
Let $C_t$ be the first set in this sequence that satisfies $C_{t}=\O$ or $C_{t}=\{0\}$. Define $X=\{j_i|i=0,1,\ldots,t-1\}$.  Then $X$ is a hitting set for $A+h_0$ of size $t$.
Therefore, by (\ref{tempooo}) we have
$$\ETD(A) \le \HS(A+h_0)\le t\le \frac{\ln(|A|-1)}{\ln\left(1-\frac{1}{\DEN(A)}\right)^{-1}}+1\le \DEN(A)\cdot\ln|A|+1.$$
\qed
\end{proof}

We now give proof sketch of
\begin{lemma} There is a set $A\subseteq B_m$ of size $n$ where $m=poly(n)$ such that $\ETD(A)=\Omega(\log n)$ and $\DEN(A)=O(1).$
\end{lemma}
\begin{proof} Consider a random uniform set $A\subseteq B_m$ of size $n$.
The probability that there are $k=(\log n)/2$ entries $i_1,\ldots,t_k\in [m]$
such that no $a\in A$ satisfies $a_{i_1}=a_{i_2}=\cdots=a_{i_k}=0$ is
$${m\choose k}\left( 1-\frac{1}{2^k}\right)^n\le \frac{1}{4}.$$
Therefore, with probability at least $3/4$, $\ETD z(A)\ge k$ and then $\ETD(A)=\Omega(\log n)$.

The probability that some subset $B\subseteq A$ of size $|B|>100$ has $\MAMI(B)\le |B|/100$
is at most
$$2^n \left(\frac{1}{2}\right)^m<\frac{1}{4}.$$
Therefore with probability at least $3/4$, $\MIMA(B)\ge |B|/100$ and $\DEN(A)=O(1)$.\qed
\end{proof}

\section{Appendix B}
In this appendix we give the proof of Lemma~\ref{LBo2} that is the same as the proof of Lemma~3.2 in~\citep{H95}.

\noindent{\bf Lemma~\ref{LBo2}}~\citep{M83,H95} Let $A\subseteq \{0,1\}^m$ of size $n$. Then
$$\OPT(A)\le \ETD(A)+\frac{\ETD(A)}{\log \ETD(A)}\log n\le \frac{2\cdot \ETD(A)}{\log \ETD(A)}\log n.$$

\begin{proof} Consider the algorithm in Figure~\ref{FCC}.
In Step 3, the algorithm defines a hypothesis that is the bitwise majority
of all the vectors in $\A^{(i,1)}$. In Step 7 an index~$y$
is found that maximizes the size of
$$\A^{(i,k)}_{(y,f_y)}:=\left\{g\in \A^{(i,k)}\ |\ g_y=f_y\right\}.$$
Suppose the variable $i$ (in the algorithm) gets the values  $1,2,\ldots,t+1$
and for each $1\le i\le t$ the variable $k$ gets the values $0,1,2,\ldots,k_i$.
Then the number of membership queries asked by the algorithm
is $k_1+\cdots+k_t$. We first prove the following

\noindent {\bf Claim} {\em For $i=1,\ldots,t-1$ we have
$$|\A^{(i+1,1)}|\le \frac{|\A^{(i,1)}|}{\max(2,k_i)}.$$}

\begin{proof} Since $S$ is a specifying set for $h$,
either some $y\in S$ satisfies $h_y\not= a_y$ or
$a$ is the only column in $\A$ that is consistent with $h$ on $S$.
Therefore, since $h=$ Majority$(\A^{(i,1)})$, we have
\begin{eqnarray}\label{fi}
|\A^{(i+1,1)}|\le \frac{|\A^{(i,1)}|}{2}.
\end{eqnarray}

Let $D=\A^{(i,1)}$ and $D'=\A^{(i+1,1)}$. Suppose $y_1,\ldots,y_{k_i}$
are the queries that were asked in the $i$th stage and let $\delta_j=a_{y_j}$
for $j=1,\ldots,k_i$. Then
$$D'=D_{(y_1,\delta_1),(y_2,\delta_2),\ldots, (y_{k_i},\delta_{k_i})}$$
and (disjoint union)
$$D=D_{(y_1,\bar\delta_1)}\cup D_{(y_1,\delta_1),(y_2,\bar\delta_2)}
\cup \cdots\cup D_{(y_1,\delta_1),(y_2,\delta_2)\ldots,(y_{k_i-1},\delta_{k_i-1}),(y_{k_i},\bar\delta_{k_i})}\cup D'.$$
Let $D^{(j)}= D_{(y_1,\delta_1),(y_2,\delta_2)\ldots,(y_{j},\delta_{j})}$, the set of columns in $D$
that are consistent
with the target column on the first $j$ assignments $y_1,\ldots,y_{j}$.
Then
$$D=D^{(0)}_{(y_1,\bar\delta_1)}\cup D^{(1)}_{(y_2,\bar\delta_2)}
\cup \cdots\cup D^{(k_i-1)}_{(y_{k_i},\bar\delta_{k_i})}\cup D'.$$
For $0\le j\le k_i-2$, the fact that we took $y_{j+1}$ for the
$(j+1)$th query and not $y_{k_i}$ implies that
$|D^{(j)}_{(y_{j+1},h_{y_{j+1}})}|\le |D^{(j)}_{(y_{k_j},h_{y_{k_j}})}|$. Therefore, for $0\le j\le k_i-2$
$$|D^{(j)}_{(y_{j+1},\overline{\delta_{j+1}})}|=|D^{(j)}_{(y_{j+1},\overline{h_{y_{j+1}}})}|\ge
|D^{(j)}_{(y_{k_j},\overline{h_{y_{k_j}}})}|=|D^{(j)}_{(y_{k_j},\delta_{k_j})}|\ge |D'|.$$
Therefore
$$|D|=|D^{(0)}_{(y_1,\bar\delta_1)}|+|D^{(1)}_{(y_2,\bar\delta_2)}
|+ \cdots +|D^{(k_i-1)}_{(y_{k_i},\bar\delta_{k_i})}|+| D'|\ge k_i\cdot |D'|.$$
With (\ref{fi}), the result of the claim follows.\qed
\end{proof}
Let $z_i=\max(2,k_i)$. Then
$$1\le |\A^{(t,1)}|\le \frac{n}{\prod_{i=1}^t z_i}$$ and therefore
$\sum_{i=1}^{t-1} \log z_i\le \log n$.  Now for $E\ge 4$ and since $E\le n$
\begin{eqnarray*}
\sum_{i=1}^t k_i&=&k_t+\sum_{i=1}^{t-1}\log z_i\frac{k_i}{\log z_i}\\
&\le& k_t+\max_i\frac{k_i}{\log z_i}\log n \\
&\le& E+\frac{E}{\log E}\log n\le  \frac{2E}{\log E}\log n.\\
\end{eqnarray*}
It is also easy to show that the above is also true for $E=2,3$.

We now prove the time complexity. Finding a specifying set at each
iteration of the While loop takes time $T$ and the number of iterations in at most $\log n$.
This takes $T\log n$ time. Now at the first iteration we define an array
of length $|S|\le E$ that contains $|\A^{(i,1)}_{(z,h_z)}|$ for each $z\in S$.
This takes at most $|\A^{(i,1)}|\cdot E$ time. Now if we
have such array for $\A^{(i,k)}_{(z,h_z)}$, we can find $y$ (in Step 7) in
time $E$ and update the array for $\A^{(i,k+1)}=\A^{(i,k)}_{(y,h_y)}$ in time
$|\A^{(i,k)}\backslash \A^{(i,k)}_{(y,h_y)}|\cdot E$. Therefore the
time of the Repeat loop is at most $2|\A^{(i,1)}|\cdot E$. Since $|\A^{(i+1,1)}|\le |\A^{(i,1)}|/2$,
the time of the While loop is at most $4n\cdot E$. This gives the result.
\qed
\end{proof}
\begin{figure}
\begin{center}
\begin{small}
\noindent{Algorithm: Find the hidden column $a\in \A$.}
\begin{tabbing}
XXXXXXXXXXXXXX\=xxx\=xxxxx\=xxxxx\=xxxxx\=xxxxx \kill
\> \ 1. $i\gets 1$, $k\gets 0$, $\A^{(1,1)} \gets \A$.\\
\> \ 2. While $|\A^{(i,1)}|\ge 2$ do \\
\>\ 3.\> $h\gets$ Majority$(\A^{(i,1)})$\\
\>\ 4.\> Find a specifying set $S$ for $h$ with respect to $\A^{(i,1)}$ \\
\>\ 5.\> Repeat \\
\>\ 6.\>\> $k\gets k+1$.\\
\>\ 7.\>\> Find $y\gets \argmin_{z\in S}\left|\A^{(i,k)}_{(z,h_z)}\right|$\\
\>\ 8.\>\> Ask query ``What is $a_y$''?\\
\>\ 9.\>\> $\A^{(i,k+1)}\gets \A^{(i,k)}_{(y,a_y)}$ \\
\>10.\>\> $S\gets S\bs \{y\}$.\\
\>11.\> Until ($h_y\not= a_y$ or $|\A^{(i,k+1)}|=1$)\\
\>12. $\A^{(i+1,1)}\gets \A^{(i,k+1)}$, $i\gets i+1$, $k\gets 0$\\
\>13. End While\\
\>14. Output the column in $\A^{(i,k)}$.
\end{tabbing}
\end{small}
\caption{\sl An algorithm that find the hidden column $a\in \A$ } \label{FCC}
\end{center}
\end{figure}
We now give another proof
\begin{proof}{\bf of Theorem~\ref{main}}
Consider the following algorithm.
After the $i$th query, the algorithm defines a set $\A_i\subseteq \A$
of all the columns that are consistent with the answers of the queries that were asked so far.
Consider any $0<\epsilon<1$.
Now the algorithm searches for a $j\in [m]$
such that $$\epsilon |\A_i|\le |\{a\in \A_i\ |\ a_j=0\}|\le (1-\epsilon) |\A_i|.$$
If such $j\in [m]$ exists then the algorithm asks ``What is $a_j$?''. Let the answer be~$\xi$.
Define $\A_{i+1}=\{a\in \A_i\ |\ a_j= \xi\}$. Obviously,
in that case,
$$|\A_{i+1}|\le (1-\epsilon) |\A_i|.$$

If no such $j\in [m]$ exists then the algorithm finds a specifying set $T_h$ for
$h:={\rm Majority}(\A_i)$, where ``Majority'' is the bitwise majority function. Then
asks queries ``What is $a_j$'' for all $j\in T_h$. If the answers are consistent
with $h$ on $T_h$ then there is a unique column $c\in \A_i$ consistent
with the answers and the algorithm outputs the index of this column.
Otherwise, there is $j_0\in T_h$ such that $a_{j_0}\not= h_{j_0}$. It is easy to see that
in that case
$$|\A_{i+1}|\le \epsilon |\A_i|.$$
Now when $\epsilon=\ln E/E$ we get
\begin{eqnarray*}
\OPT(A)&\le& \max\left( E\left\lceil\frac{\log n}{\log(1/\epsilon)}\right\rceil,
\left\lceil\frac{\log n}{\log(1/(1-\epsilon))}\right\rceil\right) \\
&\le & \frac{2E}{\log E} \log n.
\end{eqnarray*}
The time complexity of this algorithm is $O(T\log n+mn)$.\qed
\end{proof}
In fact one can prove the bound
$$\OPT(A)\le \left(\frac{E}{\log E}+\frac{E\log\log E}{\log^2E}+o\left(\frac{E\log\log E}{\log^2E}\right)\right)\log n$$
by substituting $\epsilon=(\ln E)/(E(1+\ln\ln E/\ln E)).$

\section{Appendix C}
In this Appendix we find $\ETD(\FF_\vee)$ exactly. We prove
$$\ETD(\FF_\vee)=\max_{G\in G(\FF_\vee)} |\De(G)|+\HS(\As(G)\wedge \bar{G}).$$

The following result is from~\citep{BDVY17}.

\begin{lemma}\label{uniqwit} Let $\De(G)=\{G_1,G_2,\ldots,G_t\}$ be the set of immediate descendants of $G$. If $a$ is a witness for $G_1$ and $G$, then $a$ is not a witness for $G_i$ and $G$ for all $i>1$. That is, $G_1(a)=0$, $G(a)=1$, and $G_2(a)=\cdots=G_t(a)=1$.
\end{lemma}

\subsection{Teaching Dimension}

The minimum size of a witness set for $G$ in $C$ is called the {\it witness size} and is denoted by $\TD(C,G)$. The value $$\TD(C):=\max_{G\in C}\TD(C,G)$$ is called the {\it teaching dimension} of $C$,~\citep{GK95,GRS89,SM90}.
Obviously,
$$\ETD(C,G)\ge \TD(C,G),\mbox{\ \ \ \ and}\ \ \ \ \ETD(C)\ge \TD(C).$$

\subsection{The Proof}
\begin{lemma}\label{llll} For every $G\in \FF_\vee$ we have
$$\TD(\FF_\vee,G)\ge |\De(G)|+\HS(\As(G)\wedge \bar{G}).$$
In particular,
$$\ETD(\FF_\vee)=\TD(\FF_\vee)=\max_{G\in G(\FF_\vee)} \left(|\De(G)|+\HS(\As(G)\wedge \bar{G})\right).$$
\end{lemma}
\begin{proof} Let $B$ be a witness set for $G$ in $\Ne(G)$. Take any $G'\in \De(G)$. Then there is $a\in B$ such that $G'(a)=0$ and $G(a)=1$. Since for any ascendent $G''$ of $G$ we have $G''(a)=1$, $a$ is not a witness to $G$ and any of its ascendants. By Lemma~\ref{uniqwit}, $a$ cannot be a witness to any other descendent. In the similar way, a witness for an ascendent of $G$ and $G$ cannot be a witness for any descendent of $G$ and $G$. Therefore,
\begin{eqnarray}\label{lala}
\TD(\FF_\vee,G)\ge \TD(\Ne(G),G)&= &\TD(\De(G),G)+\TD(\As(G),G)\nonumber\\
&=&|\De(G)|+\TD(\As(G),G).
\end{eqnarray}
Now let $S$ be a witness set for $G$ in $\As(G)$. Then for every $G''\in \As(G)$ there is $a\in S$ such that $G''(a)=1$ and $G(a)=0$ which is equivalent to $G''(a)\wedge \bar{G}(a)=1$. Therefore,
$$\TD(\As(G),G)\ge \HS(\As(G)\wedge \bar{G}).$$
This with (\ref{lala}) gives the result.\qed
\end{proof}

\section{Appendix D}

\subsection{Example of Classes}
\ignore{Take the class $\FF=\{f_a\ |\ a=1,2,\ldots,t\}$
where $f_a:\Re\to \{0,1\}$ is the ray $f_a(x)=[x\ge a]$, i.e.,  $f_a(x)=1$ if $x\ge a$ and $0$ otherwise. The equivalence classes $\FF_\vee^*=\{[f_a]\ |\ a=1,2,\ldots,t\}$. The representative element of $[f_a]$ is $f_a\vee f_{a+1}\vee \cdots\vee f_{t}$. We have $|\FF_\vee^*|=t$ and the learning algorithm is equivalent to binary search that asks $\lceil\log |\FF_\vee^*|\rceil=\lceil \log t\rceil$ membership queries. It is easy to show that $\OPT(\FF_\vee)= \lceil\log |\FF_\vee^*|\rceil.$}
Define the class $\Ray_n^m$. The functions are $f_{i_1,i_2,\ldots,i_m}(x_1,\ldots,x_m):[n]^m\to\{0,1\}$ where $f_{i_1,i_2,\ldots,i_m}(x_1,\ldots,x_m)=\bigwedge_{j=1}^m[x_j\ge i_j]$.
It is easy to see that this class contains $O(n^m)$ functions and its Hasse degree is $2m$.
See $\Ray_4^2$ in Figure~\ref{HasseRay24}.

See figure~\ref{RAY23E} for another example of $\FF$ with Hasse degree $3$.

\begin{figure}
\centering
\includegraphics[trim = 0 1cm 0 1cm,width=0.9\textwidth]{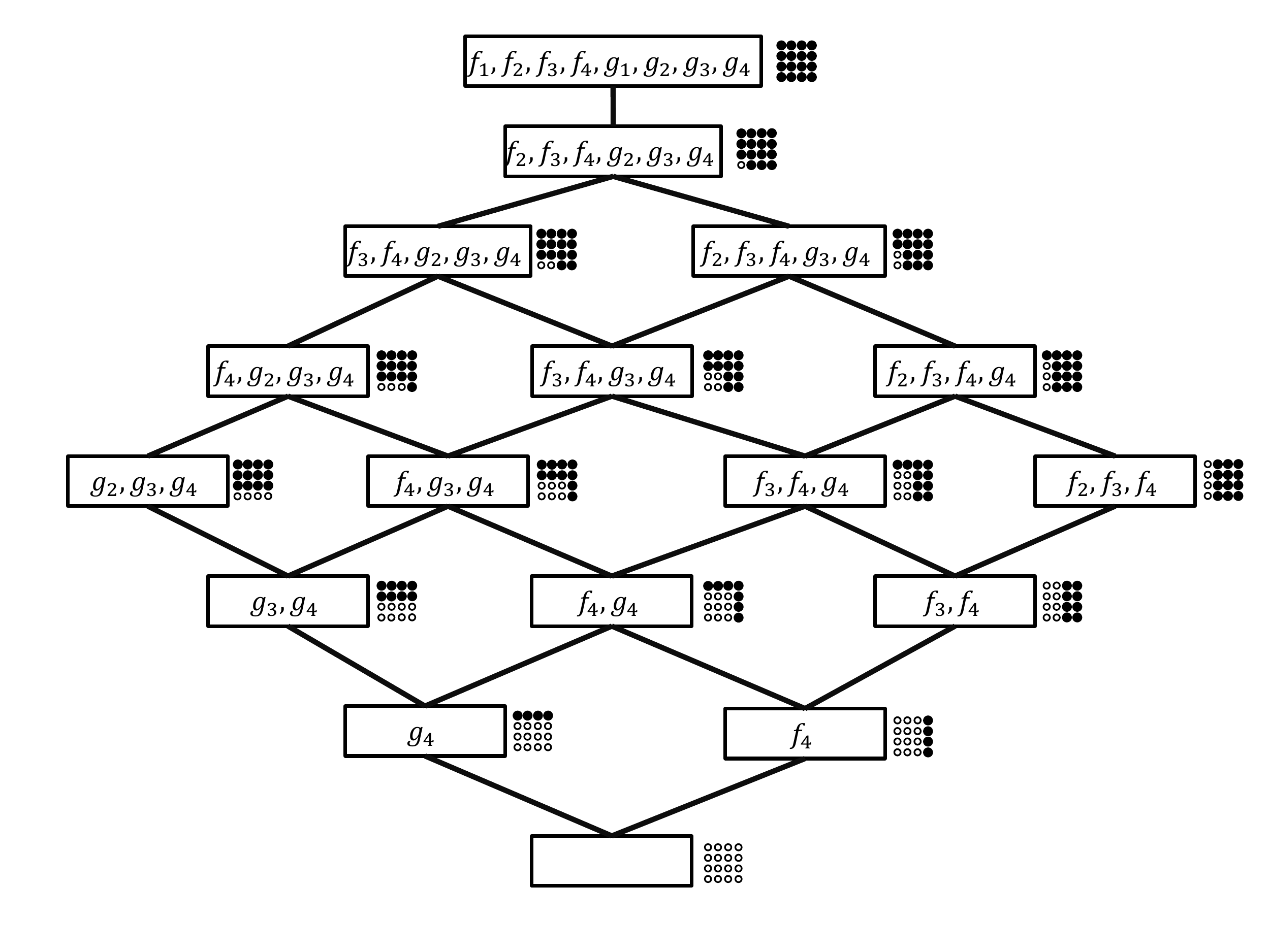}
\caption{Hasse diagram of $\Ray_4^2$. The functions are $f_i(x_1,x_2)=[x_1\ge i]$ and $g_i(x_1,x_2)=[x_2\ge i]$.}
\label{HasseRay24}
\end{figure}

\begin{figure}
\centering
\includegraphics[trim = 0 1cm 0 1cm,width=0.9\textwidth]{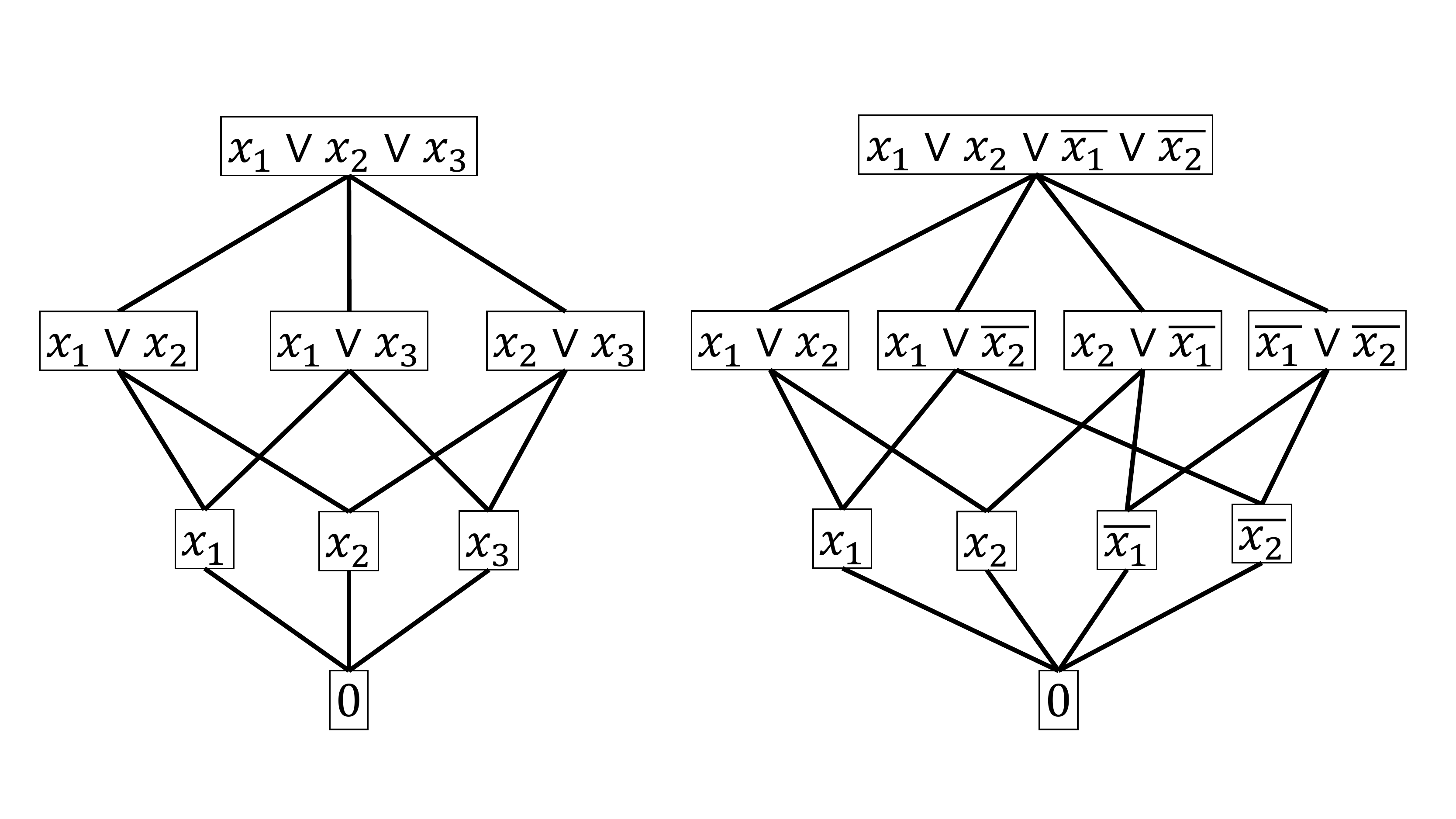}
\caption{Hasse diagram of ...}
\label{HasseClause32}
\end{figure}

\begin{figure}
\centering
\includegraphics[trim = 0 1cm 0 1cm,width=0.9\textwidth]{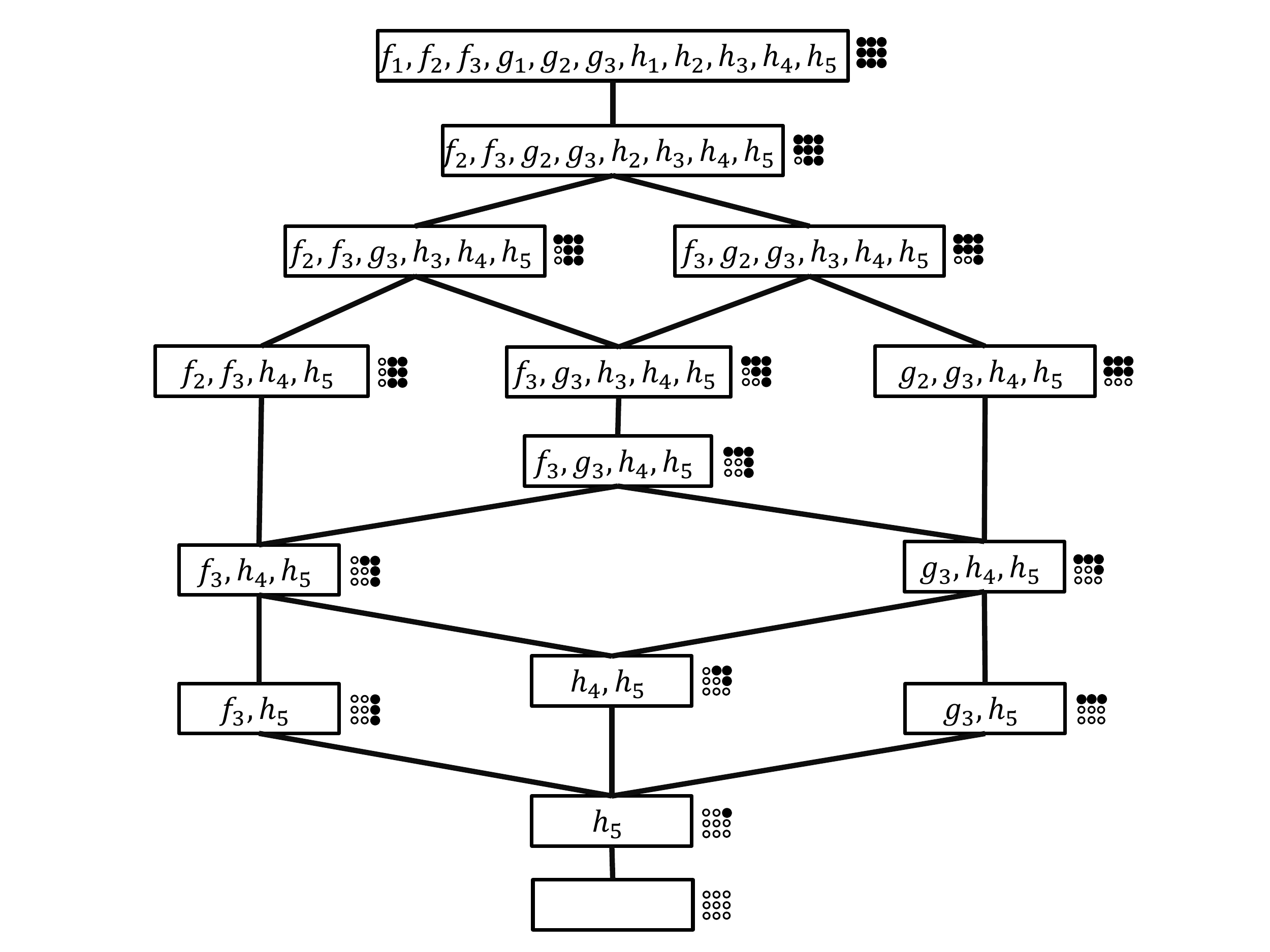}
\caption{Hasse diagram when $\FF=\{f_1,f_2,f_3,g_1,g_2,g_3,h_1,\ldots,h_5\}$ of functions $\{1,2,3\}\times\{1,2,3\}\to \{0,1\}$ where $f_i(x_1,y_1)=[x_1\ge i]$, $g_i(x_1,x_2)=[x_2\ge i]$ and $h_i(x_1,x_2)=[x_1+x_2\ge i+1]$.}
\label{RAY23E}
\end{figure}

\end{document}